\documentclass[11pt,a4paper]{llncs}
\usepackage[left=30mm,right=30mm,top=32mm,bottom=35mm]{geometry}

\usepackage{lineno}
\usepackage{latexsym,amsmath,amssymb}
\usepackage{graphicx,epsfig}
\usepackage{algo}
\usepackage{url}
\usepackage{multirow}
\usepackage[nocompress]{cite}

\def\swap{\textit{swap}}
\def\polylog{\mathrm{polylog}}

\def\IN{{\mathbb N}}

\renewcommand{\epsilon}{\varepsilon}

\newcommand{\PNF}{\mathrm{PNF}}
\newcommand{\LPN}{{\mathcal L}_{\textrm{PN}}}

\newtheorem{observation}{Observation}

\begin{document}

%%%%%%%%%%%%%%%%%%%%%%%%%%%%%%%%%%%%%%%%%%%%%%%%%%%%%%%%%%%%%%%%%%%%%%%%%%%%%%%%%%%
\title{On Combinatorial Generation of Prefix Normal Words}
%%%%%%%%%%%%%%%%%%%%%%%%%%%%%%%%%%%%%%%%%%%%%%%%%%%%%%%%%%%%%%%%%%%%%%%%%%%%%%%%%%%

\author{P\'eter Burcsi\inst{1} \and Gabriele Fici\inst{2} \and Zsuzsanna Lipt\'ak\inst{3} \and Frank Ruskey\inst{4} \and Joe Sawada\inst{5}}

\institute{Dept.\ of Computer Algebra,
%Faculty of Informatics, % this may be omitted
E\"otv\"os Lor\'and Univ., Budapest, Hungary, 
\email{bupe@compalg.inf.elte.hu}
\and Dip.\ di Matematica e Informatica, University of Palermo, Italy, 
\email{gabriele.fici@math.unipa.it} 
\and Dip.\ di Informatica, University of Verona, Italy, 
\email{zsuzsanna.liptak@univr.it}
\and Dept.\ of Computer Science, University of Victoria, Canada, 
\email{ruskey@cs.uvic.ca}
\and School of Computer Science, University of Guelph, Canada, 
\email{jsawada@uoguelph.ca}
}

\date{}

\maketitle

\begin{abstract}
A prefix normal word is a binary word with the property that no substring has more $1$s than the prefix of the same length. This class of words is important in the context of binary jumbled string matching. In this paper we present an efficient algorithm for exhaustively listing the prefix normal words with a fixed length. The algorithm is based on the fact that the language of prefix normal words is a bubble language, a class of binary languages with the property that, for any word $w$ in the language, exchanging the first occurrence of $01$ by $10$ in $w$ results in another word in the language. We prove that each prefix normal word is produced in $O(n)$ amortized time, and conjecture, based on experimental evidence, that the true amortized running time is $O(\polylog(n))$. 
\end{abstract}

%%%%%%%%%%%%%%%%%%%%%%%%%%%%%%%%%%%%%%%%%%%%%%%%%%%%%%%%%%%%%%%%%%%%%%%%%%%%%%%%%%%
\section{Introduction}
%%%%%%%%%%%%%%%%%%%%%%%%%%%%%%%%%%%%%%%%%%%%%%%%%%%%%%%%%%%%%%%%%%%%%%%%%%%%%%%%%%%

A binary word of length $n$ is {\em prefix normal} if for all $1 \leq k \leq n$, no substring of length $k$ has more $1$s than the prefix of length $k$. For example, $1001010$ is not prefix normal because the substring $101$ has more $1$s than the prefix $100$. These words were introduced in~\cite{FL11}, where it was shown that each binary word $w$ has a canonical {\em prefix normal form} $w'$ of the same length: $w$ and $w'$ are equivalent in a certain sense. 

The study of prefix normal words, and prefix normal forms, is motivated by the string problem known as {\em binary jumbled pattern matching}. In that problem, we are given a text of length $n$ over a binary alphabet, and two numbers $x$ and $y$, and ask whether the text has a substring with exactly $x$ $1$s and $y$ $0$s. While the online version can be solved with a simple sliding window algorithm in $O(n)$ time, the offline version, where many queries are expected, has recently attracted much interest: here an index of size $O(n)$ can be generated which then allows answering queries in constant time~\cite{CFL09}. However, the best construction algorithms up to date have running time O$(n^2/\log n)$~\cite{BCFL10,MR10}. Several recent papers have yielded better results under specific assumptions, such as word level parallelism or highly compressible strings~\cite{MR12,BFKL13,GG13,CGGLLRT13}, or for constructing an approximate index~\cite{CLWY12}; but the general case has not been improved. It was demonstrated in~\cite{FL11,BFKL13} that prefix normal forms of the text can be used to construct this index. See the Appendix for a brief explanation of this connection.

{\em Jumbled Pattern Matching} (JPM), over an arbitrary alphabet, is a variant of approximate pattern matching: We are given a text and a pattern, and want to answer the question whether the text has a substring which is a permutation of the pattern (existence), or find one or all occurrences of such substrings (occurrence, listing).%
\footnote{Formally: the Parikh vector of a string $s$ over a finite ordered alphabet $\Sigma = \{a_1,\ldots, a_{\sigma}\}$ is the vector $p(s) = (p_1,\ldots, p_{\sigma})$ s.t.\ for all $1\leq i \leq \sigma$, $p_i$ is the number of occurrences of $a_i$ in $s$. Given the text $s$ and pattern $t$, we want to find (occurrences of) substrings $t'$ of $s$ s.t.\ $p(t') = p(t)$.}
This problem has also been studied under the terms {\em Abelian pattern matching, Parikh vector matching, and permutation matching}. A closely related problem is that of Parikh fingerprints~\cite{AALS03}.
Applications in computational biology include SNP discovery, alignment, gene clusters, pattern discovery, and mass spectrometry data interpretation~\cite{Boecker07,Benson03,BoeckerJMS08,DuhrkopLMB13,Parida06}. 

For one query, the JPM problem can be solved in optimal linear time with a classical sliding window approach~\cite{BEL04}, while recently, interest has turned towards the indexing problem~\cite{CFL09,KRR13}. Moreover, several variants of the original problem have recently been introduced: approximate JPM \cite{BCFL12_TOCS}, JPM in the streaming model \cite{LLZ12}, JPM on trees and graphs \cite{GHLW13,CGGLLRT13}.

{\em Bubble languages} are an interesting new class of binary languages defined by the following property: ${\cal L}$ is a bubble language if, for every word $w \in {\cal L}$, replacing the first occurrence of $01$ (if any)  by $10$ results in another word in ${\cal L}$~\cite{RSW12,Ruskey12,SW12}. A generic generating algorithm for bubble languages was given in~\cite{SW12}, leading to Gray codes for each of these languages, while the algorithm's efficiency depends only on a language-dependent subroutine. In the best case, this leads to CAT (constant amortized time) generating algorithms. Many important languages are bubble languages, among them necklaces, binary Lyndon words, and $k$-ary Dyck words. %

%%%%%%%%%%%%%%%%%%%%%%%%%%%%%%%%%%%%%%%%%%%%%

In this paper, we show that prefix normal words form a bubble language and present an efficient generating algorithm which runs in $O(n)$ amortized time per word, and which yields a Gray code for prefix normal words. The best previous generating algorithm for prefix normal words ran in $O(2^n \cdot n^2)$ time, and consisted simply in testing each binary word for the prefix normal property (unpublished). Based on experimental evidence, we conjecture that the running time of our algorithm is in fact  $\Theta(\polylog(n))$ amortized. We also give a new characterization of bubble languages in terms of a closure property in the computation tree of a certain generating algorithm for binary words. 
%For the generating algorithm, we prove new properties of prefix normal words and present a testing algorithm which runs in linear time under certain conditions. 
We prove new properties of prefix normal words and present a linear time testing algorithm for words which have been obtained from prefix normal words via a simple operation. This could lead to a better understanding of prefix normal words and in the long run, to faster  computation of prefix normal forms, and thus contribute to the goal of strongly subquadratic algorithms for binary jumbled pattern matching.

%We note that some technical proofs are omitted from the main text and can be found in the Appendix.
%Some proofs have been omitted for lack of space and can be found in the Appendix.

%%%%%%%%%%%%%%%%%%%%%%%%%%%%%%%%%%%%%%%%%%%%%%%%%%%%%%%%%%%%%%%%%%%%%%%%%%%%%%%%%%%
\section{Basics}
%%%%%%%%%%%%%%%%%%%%%%%%%%%%%%%%%%%%%%%%%%%%%%%%%%%%%%%%%%%%%%%%%%%%%%%%%%%%%%%%%%%

A {\em binary word} (or {\em string}) $w=w_1\cdots w_n$ over $\Sigma=\{0,1\}$ is a finite sequence of elements from $\Sigma$. Its length $n$ is denoted by $|w|$. For any $1\leq i\leq |w|$, the $i$-th symbol of a word $w$ is denoted by $w_{i}$. 
We denote by $\Sigma^n$ the words over $\Sigma$ of length $n$, and by $\Sigma^{*} = \cup_{n\geq 0} \Sigma^n$ the set of finite words  over $\Sigma$. The empty word is denoted by $\epsilon$. 
Let $w\in \Sigma^{*}$. If $w=uv$ for some $u,v\in\Sigma^{*}$, we say that $u$ is a \emph{prefix} of $w$ and $v$ is a \emph{suffix} of $w$. A \emph{substring} of $w$ is a prefix of a suffix of $w$. A {\em binary language} is any subset $\cal L$ of $\Sigma^*$.

In the following, we will often write binary words $w\neq 1^n$ in a canonical form $w=1^s0^t\gamma$, where $\gamma\in 1\{0,1\}^* \cup \{\epsilon\}$ and $s\geq 0, t\geq1$. In other words, $s$ is the length of the first, possibly empty, $1$-run of $w$, $t$ the length of the first $0$-run, and $\gamma$ the remaining, possibly empty, suffix. Note that this representation is unique. We call $1^s0^t$ the {\em critical prefix} of $w$ and $cr(w)=s+t$ the {\em critical prefix length} of $w$. We denote by $|w|_c$ the number of occurrences in $w$ of character $c\in\{0,1\}$, and by ${\cal B}^n_d$ the set of all binary strings $w$ of length $n$ such that $|w|_1 = d$ (the  {\em density} of $w$ is $d$). 

Given a string $w$, we can obtain another string $w' = \swap(w,i,j)$, the string obtained from $w$ by exchanging the characters in positions $i$ and $j$. 

\subsection{Prefix Normal Words}\label{sec:pnw}

Let $w\in \Sigma^*$. For $i=0,\ldots,n$, we set 
\begin{itemize}
\item $P(w,i) = |w_1\cdots w_i|_1$, the number of $1$s in the $i$-length prefix of $w.$
\item $F(w,i) = \max \{|u|_1 : u \text{ is a substring of } w \text{ and } |u|=i\}$, the maximum number of $1$s over all substrings of length $i$. 
\end{itemize}

\begin{definition} 
A binary word $w$ is {\em prefix normal} if, for all $1\leq i \leq |w|$, $F(w,i) = P(w,i)$. In other words, a word is prefix normal if no substring contains more $1$s than the prefix of the same length.
\end{definition}

For example, $1001010$ is not prefix normal because the substring $101$ has more $1$s than the prefix $100$. We denote by $\LPN$ the language of prefix normal words. %For more details, see~\cite{FL11}. 
In~\cite{FL11} it was shown that for every word $w$ there exists a unique word $w'$, called its {\em prefix normal form}, or $\PNF(w)$, such that for all $1\leq i \leq |w|$, $F(w,i) = F(w',i)$, and $w'$ is prefix normal. Therefore, a prefix normal word is a word coinciding with its prefix normal form.

{\em Note:} In~\cite{FL11}, the property `prefix normal' was defined both with respect to $1$ and with respect to $0$. Here we restrict ourselves to prefix normal words w.r.t.\ $1$.

\begin{table}[ht]
\centering \begin{small}
\begin{raggedright}

\begin{tabular}{*{2}l @{\hspace{6mm}}||@{\hspace{6mm}} *{2}l}
$\LPN \cap \Sigma^5$ \quad  & words with this PNF  & $\LPN \cap \Sigma^5$  \quad  & words with this PNF\\
\hline &&&\rule[-2pt]{0pt}{3pt}\\
$11111$ & \{$11111$\} & $11000$ & \{$11000,011000,00110,00011$\}\\
$11110$ & \{$11110$, $011111$\}& $10101$ & $\{10101\}$\\
$11101$ & \{$11101$, $10111$\}& $10100$ & $\{10100, 01010, 00101\}$\\
$11100$ & \{$11100$, $01110$, $00111$\}& $10010$ & $\{10010, 01001\}$\\
$11011$ & \{$11011$\}& $10001$ & $\{10001\}$\\
$11010$ & \{$11010, 10110, 01101,01011$\}& $10000$ & $\{10000, 01000, 00100, 00010, 00001\}$\\
$11001$ & \{$11001,10011$\}& $00000$ & $\{00000\}$ \\
&&&\\
\hline \vspace{4mm}
\end{tabular}
\end{raggedright}\caption{All prefix normal words of length 5 and their equivalence classes.\label{table:classes5}}
\end{small}
\end{table}

In Table~\ref{table:classes5} we list all prefix normal words of length $5$, and, for each $w'\in \LPN \cap \Sigma^5$, the set of binary words $w$ such that $\PNF(w) = w'$ (i.e., its equivalence class).
Several methods were presented in~\cite{FL11} for testing whether a word is prefix normal; however, all ran in quadratic time. One open problem given there was that of enumerating prefix normal words (counting). The number of prefix normal words of length $n$ can be computed by checking for each binary word whether it is prefix normal, i.e.\ altogether in $O(2^n\cdot n^2)$ time. In this paper, we present an algorithm that is far superior in that it generates only prefix normal words, rather than testing every binary word; it runs in $O(n)$ time per word; and it generates prefix normal words in cool-lex order, constituting a Gray code (subsequent words differ by a constant number of swaps or flips).

\subsection{Bubble Languages and Combinatorial Generation}

Here we give a brief introduction to bubble languages, mostly summarising results from~\cite{RSW12,SW12}. We also give a new characterization of bubble languages in terms of the computation tree of a generating algorithm (Obs.~\ref{obs:tree}).

\begin{definition}
A language ${\cal L} \subseteq \{0,1\}^*$ is called {\em a bubble language} if, for every word $w\in {\cal L}$, exchanging the first occurrence of $01$ (if any) by $10$ results in another word in ${\cal L}$. 
\end{definition}

%{\em Note:} In general, it is possible to define both {\em $01$-bubble languages} and {\em $10$-bubble languages} (analogously to $01$-bubble, exchanging the roles of $10$ and $01$ in the definition). Then a language is a {\em bubble language} if it is $01$- or $10$-bubble. However, in this paper, we will restrict ourselves to $01$-bubble languages. 

For example, the languages of Lyndon words, necklaces and pre-necklaces are bubble languages. 
%Gabriele: ($01$ or $10$) removed.  I commented this because it would require a detailed discussion that is superfluous.
A language ${\cal L}\subseteq \{0,1\}^n$ is a bubble language if and only if each of its fixed-density subsets ${\cal L} \cap {\mathcal B}^n_d$ is a bubble language~\cite{RSW12}. This implies that for generating a bubble language, it suffices to generate its fixed-density subsets.

Next we consider combinatorial generation of binary strings.

Let $w$ be a binary string of length $n$. Let $d$ be the number of $1$s in $w$, and let $i_1,i_2,\ldots, i_d$ denote the positions of the $1$s in $w$. 
Clearly, we can obtain $w$ from the word $1^d0^{n-d}$ with the following algorithm: first swap the last $1$ with the $0$ in position $i_d$, then swap the $(d-1)$st $1$ with the $0$ in position $i_{d-1}$ etc. Note that every $1$ is moved at most once, and in particular, once the $k$'th $1$ is moved into the position $i_k$, the suffix $w_{i_k}\cdots w_n$ remains fixed for the rest of the algorithm. 

These observations lead us to the following generating algorithm (Fig.~\ref{algo:bubble}), which we will refer to as Recursive Swap Generation Algorithm (like Alg.\ 1 from~\cite{SW12}, which in addition includes a language-specific subroutine). It generates recursively all binary strings from ${\cal B}^n_d$ with fixed suffix $\gamma$, 
where $\gamma\in 1\{0,1\}^* \cup \{\epsilon\}$, starting from the string $1^s0^t\gamma$. The call Generate($d,n-d,\epsilon$) generates all binary strings of length $n$ with density $d$.  

\begin{figure}
\begin{algorithm}{Generate($s,t,\gamma$)}{
\qcomment{current string resides in array $w$}
} 
\qif $s>0$ and $t>0$ \\ \qthen 
\qfor $i=1,2,\ldots, t$\\
\qdo 
$w \gets \swap(w,s,s+i)$ \\
{\em Generate}($s-1,i,10^{t-i}\gamma$)\\
$w \gets \swap(w,s,s+i)$
\qend
\qfi \\
{\em Visit}()
\end{algorithm}
\vspace{-4mm}
\caption{The Recursive Swap Generation Algorithm\label{algo:bubble}}
\end{figure}

The algorithm swaps the last $1$ of the first $1$-run with each of the $0$s of the first $0$-run, thus generating a new string each, for which it makes a recursive call. During the execution of the algorithm, the current string resides in a global array $w$.  In the subroutine Visit() we can, but do not have to, print the contents of this array; we may just want to increment a counter (for enumeration), or check some property of the current string. The main point of Visit() is that it touches every object once.

Let $T^n_d$ denote the recursive computation tree generated by a call to {\em Generate}($d,n-d,\epsilon$).
As an example,  Fig.~\ref{fig:example1} illustrates the computation tree $T^7_4$. 
\begin{figure}
\begin{center}
\includegraphics[width=1\textwidth]{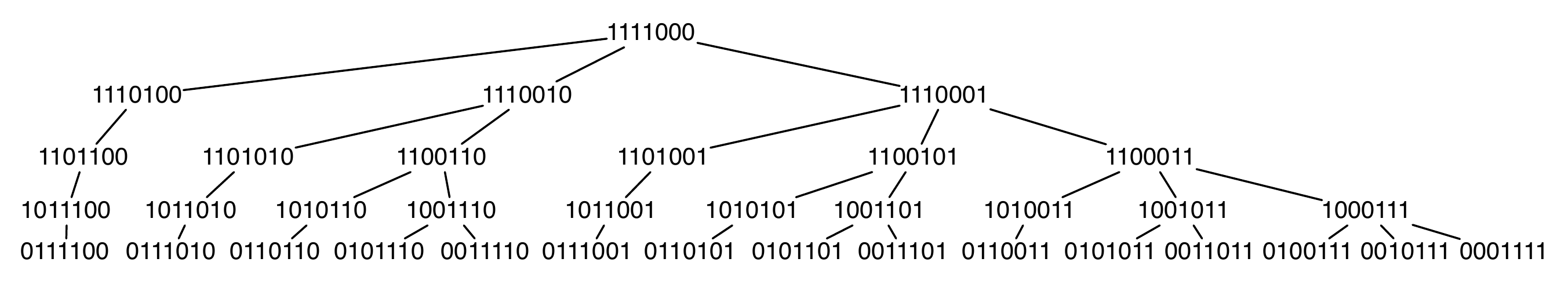}
\caption{\label{fig:example1}The computation tree $T_d^n$ for $n=7,d=4$.}
\end{center}
\end{figure}

\noindent The depth of the tree equals $d$, the number of $1$s; while the maximum degree is $n-d$, the number of $0$s. In general, for the subtree rooted at $v = 1^s0^t\gamma$, we have depth $s$ and maximum degree $t$; in particular, the number of children of $v$ is exactly $t$. In fact, $v$'s $i$th child has the form $1^{s-1}0^i10^{t-i}\gamma$. Moreover, the suffix $\gamma$ remains unchanged in the entire subtree, and the computation tree is isomorphic to the computation tree of $1^s0^t$. This $\gamma$ is called {\em fixed suffix}~\cite{RSW12}. Note also that the critical prefix length $s+t$ strictly decreases along any downward path in the tree. 

The algorithm performs a post-order traversal of the tree, yielding an enumeration of the strings of ${\cal B}_d^n$ in what is referred to as {\em cool-lex order}~\cite{Wi09,SW12,RSW12}. A pre-order traversal of the same tree, which implies moving line 4 of the algorithm before line 1, would yield an enumeration in {\em co-lex order}. A crucial property of cool-lex order is that any two subsequent strings differ by at most two swaps (transpositions), thus yielding a {\em Gray code}~\cite{RSW12}. 
This can be seen in the computation tree $T_d^n$ as follows. Note that in a post-order traversal of $T_d^n$, we have:\\

$next(u) = 
\begin{cases}
parent(u) & \text{if } u \text{ is rightmost child}\\
\text{leftmost descendant of $u$'s right sibling} & \text{otherwise.}\\
\end{cases}
$

\medskip

Let $u,u'$ both be children of $v$. This means that for some $s,t,i,j\in \IN$ and $\gamma\in 1\{0,1\}^* \cup \{\epsilon\}$, we have $v=1^s0^t\gamma$,  $u=1^{s-1}0^{i}10^{t-i}\gamma$, and $u'=1^{s-1}0^{j}10^{t-j}\gamma$. Let $v'$ be a  descendant of $v$ along the leftmost path, i.e.\ $v' = 1^k01^{s-k}0^{t-1}\gamma$ for some $k$. Then
\begin{align}\label{eq:swap}
v &= \swap(u,s,s+i) & \qquad \text{(parent)}\\
u' &= \swap(u,s+i,s+j) & \qquad \text{(sibling)} \nonumber\\
v' &= \swap(v,k,s+1) & \qquad \text{(descendant along leftmost path)}\nonumber
\end{align}

\medskip 

We now state a crucial property of bubble languages with respect to the Recursive Swap Generating Algorithm which follows immediately from the definition of bubble languages:

\begin{observation}\label{obs:tree}
A language ${\cal L}$ is a bubble language if and only if, for every $d=0,\ldots,n$, its fixed-density subset ${\cal L}\cap {\mathcal B}_d^n$ is closed w.r.t.\ parents and left siblings in the computation tree $T_d^n$ of the Recursive Swap Generating Algorithm. In particular, if ${\cal L}\cap {\mathcal B}_d^n \neq \emptyset$, then it forms a subtree rooted in $1^d0^{n-d}$.
\end{observation}

Using this property, the Recursive Swap Generating Algorithm can be used to generate {\em any} fixed-density bubble language ${\cal L}$, as long as we have a way of deciding, for a node $w=1^s0^t\gamma$, already known to be in ${\cal L}$, which is its rightmost child (if any) that is still in ${\cal L}$. If such a child exists, and it is the $k$th child $u=1^{s-1}0^k10^{t-k}\gamma$, then the bubble property ensures that all children to its left are also in ${\cal L}$. Thus, line $2.$ in the algorithm can simply be replaced by ``for $i=1,\ldots,k$''. Moreover, the Recursive Swap Generating Algorithm, which visits the words in the language in cool-lex order, will yield a Gray code, since because of this closure property, $next(u)$ will again either be the parent, or a node on the leftmost path of the right sibling, both of which are reachable within two swaps, see~\eqref{eq:swap}.

In~\cite{SW12}, a generic generating algorithm was given which moves the job of finding this rightmost child $k$ into a subroutine Oracle($s,t,\gamma$). If Oracle($s,t,\gamma$) runs in time $O(k)$, then we have a CAT algorithm. In general, this will not be possible, and a generic Oracle  tests for each child from left to right (or from right to left) whether it is in the language. Because of the bubble property, after the first negative (positive) test, it is guaranteed that no more children will be in the language, and the running time of the algorithm is amortized that of the membership tester. The crucial trick is  that it is not necessary to have a {\em general} membership tester, since all we want to know is which of the children of a node {\em already known to be in ${\cal L}$} are in ${\cal L}$; moreover, the membership tester is allowed to use other information, which it can build up iteratively while examining earlier nodes.

%%%%%%%%%%%%%%%%%%%%%%%%%%%%%%%%%%%%%%%%%%%%%%%%%%%%%%%%%%%%%%%%%%%%%%%%%%%%%%%%%%%
\section{Combinatorial Generation of Prefix Normal Words}~\label{sec:combgenpnw}
%%%%%%%%%%%%%%%%%%%%%%%%%%%%%%%%%%%%%%%%%%%%%%%%%%%%%%%%%%%%%%%%%%%%%%%%%%%%%%%%%%%

In this section we prove that the set of prefix normal words $\LPN$ is a bubble language.  Then, by providing some properties regarding membership testing, we can apply the cool-lex framework to generate all prefix normal words of a given length $n$ and density $d$ in $O(n)$-amortized time.  By concatenating the lists together for all densities in increasing order, we obtain an $O(n)$-amortized time algorithm to list all prefix normal words of length $n$.

\begin{lemma}
\label{lemma:PNbubble}
 The language $\LPN$ is a bubble language.
\end{lemma}

\begin{proof}
Let $w$ be a prefix normal word containing an occurrence of $01$. Let $w'$ be the word obtained from $w$ by replacing the first occurrence of $01$ with $10$. Then $w=u01v$, $w'=u10v$ for some $u,v\in \Sigma^{*}$. Let $z$ be a substring of $w'$. We have to show that $|z|_{1}\le P(w',|z|)$.
 
Note that for any $k$, $P(w,k)\le P(w',k)$. In fact, $P(w',|u|+1) = P(w,|u|)+1$, and for every $k\neq |u|+1$, $P(w,k) = P(w',k)$. Now if $z$ is contained in $u$ or in $v$, then $z$ is a substring of $w$, and thus $|z|_1\le P(w,|z|) \le P(w',|z|)$. If $z=u'10v'$, with $u'$ suffix of $u$ and $v'$ prefix of $v$, then  $|z|_{1}=|u'01v'|_{1}\le P(w,|z|) \le P(w',|z|)$. If $z=0v'$, with $v'$ prefix of $v$, then $|z|_{1}<|1v'|_{1}$, and $1v'$ is a substring of $w$, thus $|z|_{1}\le P((w,|z|)\le P(w',|z|)$. 
Else $z=u'1$, with $u'$ suffix of $u$. We can assume that $u'$ is a proper suffix of $u$. Let $z'$ be the substring of $w'$ of the same length as $z$ and starting one position before $z$ (in other words, $z'$ is obtained by shifting $z$ to the left by one position). Since $u$ does not contain $01$ as a substring, we have $u=1^{n}0^{m}$ for some $n\ge 1, m\ge 0$. If $z'$ is a power of $0$'s, then $|z|_{1}=1$ and the claim holds. Else, $|z|_{1}=|z'|_{1}$, and $z'$ is a substring of $w$. Thus $|z|_{1}\le P(w,|z|) \le P(w',|z|)$. \hfill \qed
\end{proof}

In Fig.~\ref{fig:example2}, we give the computation tree $T_4^7$ and highlight the subtree corresponding to $\LPN$.  Since $\LPN$ is a bubble language, by Obs.~\ref{obs:tree} it is closed w.r.t.\ left siblings and parents. However, we still have to find a way of deciding which is the rightmost child of a node that is still in $\LPN$.

\begin{figure}
\begin{center}
\includegraphics[width=1\textwidth]{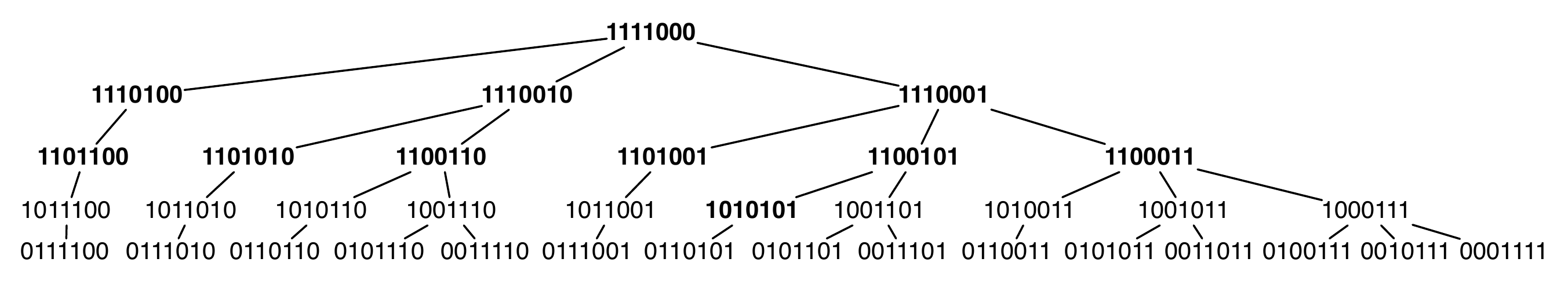}
\caption{\label{fig:example2}The computation tree $T_4^7$ for $n=7,d=4$. Prefix normal words in bold.}
\end{center}
\end{figure}

The following lemma states that, given a prefix normal word $w$, in order to decide whether one of its children in the computation tree is prefix normal, it suffices to check the PN-property for one particular length only: the critical prefix length of the child node. Moreover, this check can be done w.r.t.\ $\gamma$ only. This fact will be crucial in the generating algorithm.

\begin{lemma}\label{lemma:isPNF}
Let $w\in \LPN$, with $w=1^s0^t\gamma$, with $\gamma \in 1\{0,1\}^* \cup \{\epsilon\}$. Let $\bar\gamma = \gamma 0^{s+t}$, i.e.\ $\gamma$ padded with $0$s to length $n$. Let $w' = \swap(w,s,s+i)$. Then $w'\in \LPN$, unless one of the following holds:

\begin{enumerate}
\item $\bar\gamma$ has a substring of length $s+i-1$ with at least $s$ $1$s, or 
\item the string $w'_{s+i} \cdots w'_{2(s+i-1)}$ has at least $s$ $1$s.
\end{enumerate}

Moreover, the latter is the case if and only if $P(\bar\gamma, s + 2(i-1) - t) \geq s-1$ (where by convention, we regard a prefix of negative length as the empty word).

\end{lemma}

\begin{proof}
Let's assume that $w'\not\in \LPN$. Then there is a substring $u$ of $w'$ s.t.\ $|u|_1 > P(w',|u|)$. Let $m$ be the length of $u$.

{\em Case 1.} $u$ is a substring of $w$. Since $w\in \LPN$, therefore $|u|_1 \leq P(w,m)$. Since also $|u|_1 > P(w',m)$, this implies $s\leq m \leq s+i-1$, because for all other arguments, $P(w,\cdot)$ and $P(w',\cdot)$ coincide. Note that $u$ must have an occurrence in $w'$ which contains neither of the swapped bits, else it would not be a substring of $w$. Thus $u$ starts at some position to the right of $s+i$. Therefore we can write $u = 0^r v$, with $v$ a substring of $\gamma$; in particular, if $u$ is a substring of $\gamma$, then $r=0$ and $v=u$; otherwise, $v$ is a prefix of $\gamma$. Now set $u' = vv'$, with $v'$ the substring of $\bar\gamma$ of length $s+i-1-m+r$ following $v$. Then $u'$ has length $s+i-1$, and since $|v'|_1 \geq 0 = |0^r|_1$,  it contains at least $s$ many $1$s.

{\em Case 2.} $u$ is not a substring of $w$. Therefore it contains at least one of the two swapped bits. It cannot contain the swapped $0$ (in position $s$) because then it would be preceded only by $1$s, in which case the prefix of $w'$ of length $m$ could not have fewer $1$s than $u$. Thus, $u$ contains the swapped $1$ only (in position $s+i$). If $m>s+i-1$, then the prefix of $w'$ of length $m$ overlaps with $u$, i.e.\ we can write $w'_1\cdots w'_m = vv'$ and $u=v'u'$ for some non-empty $v'$ containing the swapped $1$. Since $|u|_1 > P(w',m)$, this implies that also $u'$ has more $1$s than the prefix of the same length. Since $u'$ is a substring of $w$, we are back in Case 1. 

So we have $m\leq s+i-1$. We can write $u=0^r10^{t-i}v$, with $v$ prefix of $\gamma$. Now remove the starting $0$s from $u$ and extend it to the right to get $u' = 10^{t-i}v'$, with $v'$ be the prefix of $\bar\gamma$ of length $k=s+2(i-2)-t$. Then $|u'| = s+i-1$ and $|u'|_1 \geq |u|_1 > P(w',|u|) = s-1$. Moreover, $|u'|_1 = 1 + P(\bar\gamma,s+2(i-1)-t)$. \hfill \qed
\end{proof}

\begin{corollary}\label{coro:isPNF}
Given $w=1^s0^t\gamma \in \LPN$. If we know $F(\gamma,\cdot)$ and $P(\gamma,\cdot)$, then it can be decided in constant time whether $w' = \swap(w,s,s+i)$ is prefix normal.
\end{corollary}

\begin{lemma}
\label{lemma:FPCalculate}
Let $\gamma' = 10^r\gamma$, with $\gamma\in 1\{0,1\}^* \cup \{\epsilon\}$. Then for all $i=0,1,\ldots,|\gamma'|$,

$$ F(\gamma',i) =  \begin{cases}\max(P(\gamma',i), F(\gamma,i)) & \text{ for } i\leq |\gamma|\\
\max(P(\gamma',i), F(\gamma,|\gamma|)) & \text{ for } i > |\gamma|. 
\end{cases}$$

\end{lemma}

\begin{proof}
A substring of length $i$ either uses the new $1$ in the first position, or it does not. If it does, then it is a prefix of $\gamma'$ and its number of $1$s is given by $P(\gamma',i)$. Else it is a substring of $0^{r-1}\gamma$, and its number of $1$s is given by $F(\gamma,i)$ for $i$ up to the length of $\gamma$, or by the number of $1$s in $\gamma$, $F(\gamma,|\gamma|)$, if $i$ spans all of $\gamma$. \hfill \qed
\end{proof}

\begin{corollary}\label{coro:F}
The $F$-function of $\gamma$ for node $w=1^s0^t\gamma$ can be computed in linear time based on the $F$-function of $w$'s parent node.
\end{corollary}

By applying these results, the algorithm {\em GeneratePN($d,n-d,\epsilon$)}  can be used to generate $\LPN \ \cap  \ {\cal B}^n_d$ in cool-lex Gray code order.  Starting from the left child and proceeding right (with respect to the computation tree $T_n^d$), the algorithm will make a recursive call until a child is no longer prefix normal. 
The membership test is done in the subroutine {\em isPN}, which uses the conditions of Lemma~\ref{lemma:isPNF}. The algorithm maintains an array $F$ which contains the maximum number of $1$s in $i$-length substrings of $\gamma$ (the $F$-function of $\gamma$), and a variable $z$. Before testing the first child, in {\em update($F,s+t$)}, it computes the current $\gamma$'s $F$-function based on the parent's (Corollary~\ref{coro:F}). Note that it is not necessary to compute all of the $F$-function, since all nodes in the subtree have critical prefix length smaller than $s+t$, thus this update is done only up to length $s+t$. After the recursive calls to the children, the array is restored to its previous state in {\em restore($F,s+t$)}. The variable $z$ contains the number of $1$s in the prefix of $\gamma$ which is spanned by the substring of case 2. of Lemma~\ref{lemma:isPNF}, for the first child. It is updated in constant time after each successful call to {\em isPN}, to include the number of $1$s in the two following positions in $\gamma$.

\begin{figure}
\begin{algorithm}{GeneratePN($s,t,\gamma$)}{
\qcomment{$w=1^s0^t\gamma$ must be prefix normal}}
\qif $s>0$ {\bf and} $t>0$ \\
\qthen {\em update}$(F,s+t)$\\
$z \gets P(\gamma,s-t)$ \\
$i \gets 1$\\
\qwhile $i \leq t$ {\bf and} {\em isPN}($\swap(w,s,s+i)$)\\
\qdo $w \gets \swap(w,s,s+i)$\\
{\em GeneratePN}($s-1,i,10^{t-i}\gamma$)\\
{\em update}$(z)$ \\
$i \gets i+1$\\
$w \gets \swap(w,s,s+i)$
\qend \\
{\em restore}($F,s+t$)
\qfi \\
{\em Visit}()
\end{algorithm}
\vspace{-4mm}
\caption{Algorithm generating all prefix normal words in the subtree rooted in $1^s0^t\gamma$.\label{algo:genPN1}}
\end{figure}

\noindent  By concatenating the lists of prefix normal words with densities $0,1,\ldots , n$, we obtain an exhaustive listing of $\LPN \ \cap  \ \Sigma^n$.

\begin{figure}
\begin{algorithm}{GeneratePN($n$)}{
\qcomment{generates all prefix normal words of length $n$}}
\qfor $d=0, 1, \ldots , n$\\
\qdo initialize $F$ of length $n$ with all $0$s\\
{\em GeneratePN($d,n-d,\epsilon$)}
\qend
\end{algorithm}
\caption{Algorithm generating all prefix normal words of length $n$.\label{algo:genPN2}}
\end{figure}

As an example, a call to {\em GeneratePN($5$)} produces the following list of prefix normal words of length 5:
\small
\[ 00000, 10000, 10100, 10010, 10001, 11000, 11010, 10101, 11001, 11100, 11011, 11101, 11110, 11111.\]
\normalsize
These strings are also given in Sec.~\ref{sec:pnw}.  Since the fixed-density listings are a cyclic Gray code (Theorem 3.1 from \cite{RSW12}), it follows that this complete listing is also a Gray code.  In fact, if the fixed-density listings are listed by the odd densities (increasing), followed by the even densities (decreasing), the resulting listing would be a cyclic Gray code.

\begin{theorem}
Algorithm {\em GeneratePN($n$)} generates all prefix normal words of length $n$ in amortized $O(n)$ time per word. 
\end{theorem}

\begin{proof}
Since $1^d0^{n-d}$ is prefix normal for every $d$, we only need to show that the correct subtrees of $T_d^n$ are generated by the algorithm. By Lemma~\ref{lemma:isPNF}, only those children will be generated that are prefix normal; on the other hand, by the bubble property (Obs.~\ref{obs:tree}), as soon as a child tests negative, no further children will be prefix normal. The running time of the recursive call on $w\in \LPN$ consists of (a) updating and restoring $F$ (lines 2 and 9): the number of steps equals the critical prefix length of $w$, which is $O(n)$; (b) computing $z$ (line 3): again $cr(w)\leq n$, the critical prefix length of $w$, many steps, so $O(n)$; and (c) work within the while-loop (lines 5 to 8), which, for a word with $k$ prefix normal children, consists of $k$ positive and $1$ negative membership tests, of $k$ updates of $z$, and the recursive calls on the positive children. The membership tests take constant time by Corollary~\ref{coro:isPNF}, so does the update of $z$. Since $w$ has $k$ prefix normal children, we charge the positive membership tests and the $z$-updates to the children, and the negative test to the current word.
So for one word $w\in \LPN$, we get $3\cdot O(n) + O(1) + 2\cdot O(1) = O(n)$ work.
\hfill \qed 
\end{proof}

%%%%%%%%%%%%%%%%%%%%%%%%%%%%%%%%%%%%%%%%%%%%%%%%%%%%%%%%%%%%%%
%%%%%%%%%%%%%%%%%%%%%%%%%%%%%%%%%%%%%%%%%%%%%%%%%%%%%%%%%%%%%%
%%%%%%%%%%%%%%%%%%%%%%%%%%%%%%%%%%%%%%%%%%%%%%%%%%%%%%%%%%%%%%

\section{Experimental results}
\newcommand{\pnn}{\ensuremath{\textit{pnw}}}

In this section we present some theoretical and numerical results about the number of prefix normal words and their structure. These have become available thanks to the algorithm presented, which allowed us to generate $\LPN$ up to length 50 on a home computer. 
Let $\pnn(n) := |\LPN \cap \Sigma^n|$. The following lemma follows from the observation that $1^{\lceil n/2 \rceil} w$ is a prefix normal word of length $n$ for all words $w$ of length $\lfloor n/2 \rfloor$.

\begin{lemma}
The number of prefix normal words grows exponentially in $n$. We have that $\pnn(n) \geq 2^{\lfloor n/2\rfloor}$. 
\end{lemma}

The first members of the sequence $\pnn(n)$ are listed in \cite{sloane2}, and these values suggest that the lower bound above is not sharp. We turn our attention to the growth rate of $\pnn(n)$ as $n$ increases. Note that $1 \leq \pnn(n)/\pnn(n-1) \leq 2$. The lower bound follows form the fact that all prefix normal words can be extended by adding a $0$ to the end, and the upper bound is implied by the prefix-closed property of $\LPN$.  Fig.~\ref{figGrowth} (left) shows the growth ratio for small values of $n$. The figure shows two interesting phenomena: the values seem to approach 2 slowly, i.e., the number of prefix normal words almost doubles as we increase the length by 1. Second, the values show on oscillation pattern between even and odd values. We have so far been unable to establish these observations theoretically.

\begin{figure}
\begin{minipage}[c]{8cm}
\begin{center}
\includegraphics[scale=0.25]{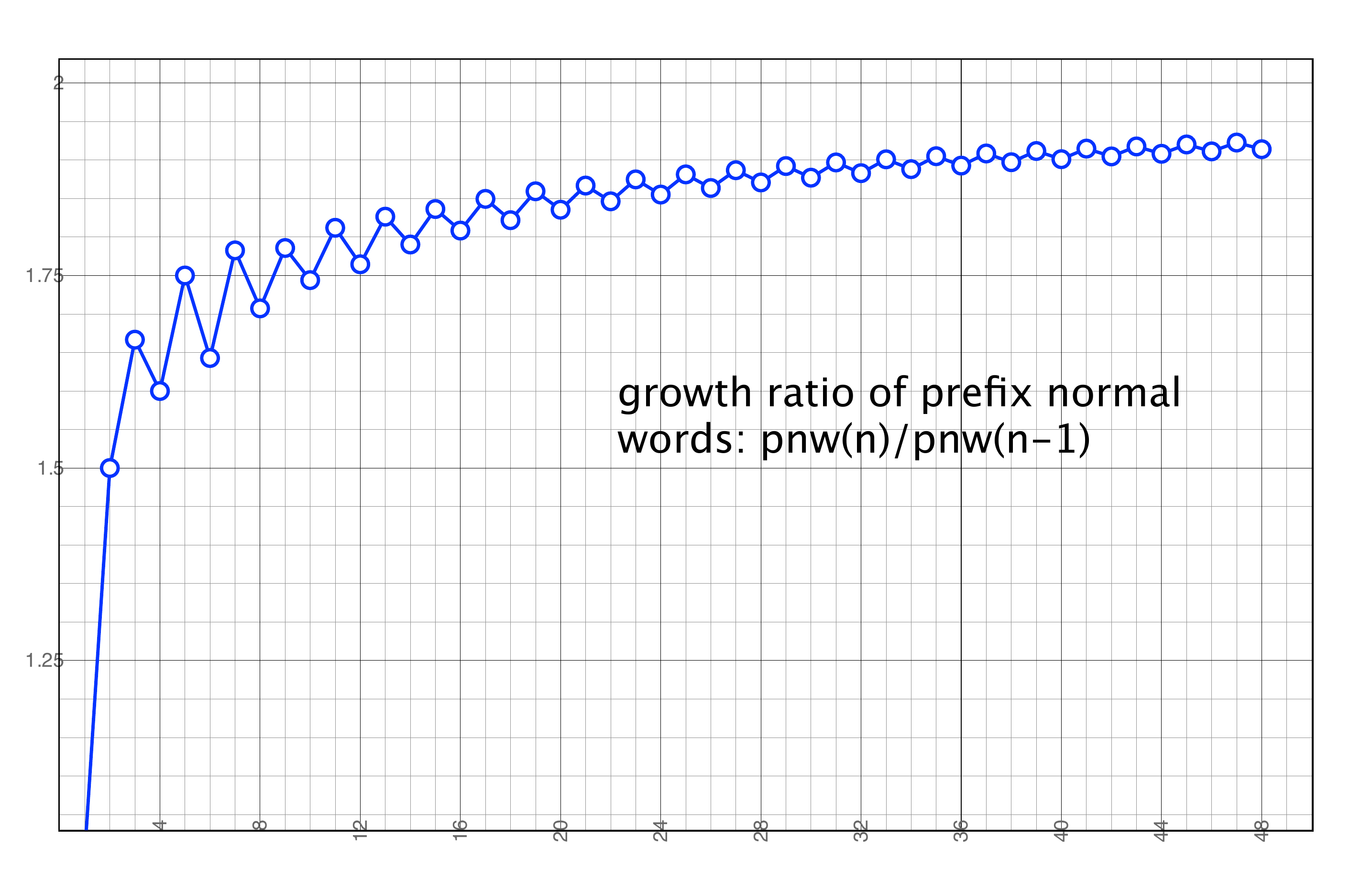}
\end{center}
\end{minipage}
\begin{minipage}[c]{8cm}
\begin{center}
\includegraphics[scale=0.25]{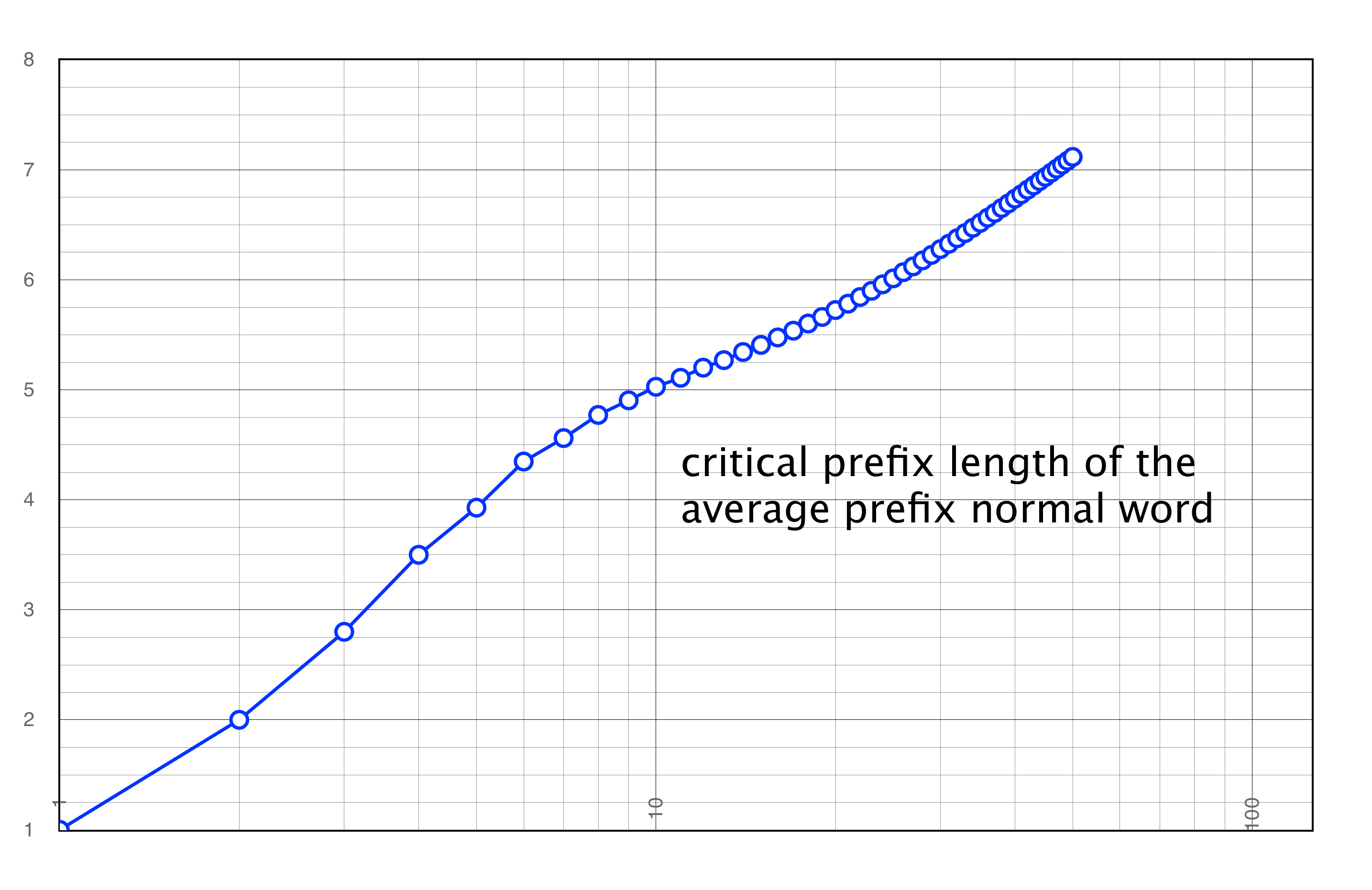}
\end{center}
\end{minipage}
\caption{The value of $\pnn(n)/\pnn(n-1)$ (left), and of $E(S+T)$ for prefix normal words for $n \leq 50$ (right).} 
\label{figGrowth}
\label{figSTGrowth}
\end{figure}

The structure of prefix normal words is also relevant for the generation algorithm, since the amortized running time of the algorithm is bounded above by the average value of the critical prefix length $s+t$ taken over all prefix normal words. This differs from the expected critical prefix length of the prefix normal form of a uniformly random word. For the latter we have the following result.

\begin{lemma}\label{lemma:random_pnf_st}
Given a random word $w$, let $w' = \PNF(w)$. Let $Z$ denote the critical prefix length of $w'$. Then for the the expected value of $Z$ we have $E(Z) = \Theta(\log n)$.
\end{lemma}

\begin{proof} 
Write $w' = 1^{s'}0^{t'}\gamma'$ in the usual form, i.e.\ with $\gamma'\in 1\{0,1\}^* \cup \{\epsilon\}$, and consider the random variables $S'=s'$ and $T'=t'$. It is known that the expected maximum length of a run of 1s in a random word is $\Theta(\log n)$\cite{GO80}. Clearly, $S'$ equals the length of the longest run of $1$s of $w$, thus $Exp(S') = \Theta(\log n)$. To determine $E(T')$, consider a $1$-run of $w$ of maximum length $s'$. If $w$ has at least another occurrence of $1$, then there is a substring of $w$ consisting of the maximal $1$-run and one more $1$; the number of $0$'s in this substring is an upper bound on $t'$. Since these $0$s form a single $0$-run, their number is again $O(\log n)$ in expectation. If on the other hand, all occurrences of $1$ in $w$ are in the maximal run, then $w' = 1^{s'}0^{n-s'},$ so $t' = n-s' \leq n$. The number of words with at most one $1$-run is ${n+1 \choose 2}+1$. So we have: 
\begin{align*}
E(S' + T') &= \Theta(\log n) + \left(1 - \frac{\Theta(n^2)}{2^n}\right)O(\log n) + \frac{\Theta(n^2)}{2^n} n = \Theta(\log n) \qquad\qed
\end{align*}
\end{proof}

The expected value of the critical prefix length for prefix normal words is shown in Fig.~\ref{figSTGrowth} (right) for $n\leq 50$, on a loglinear scale. We conjecture that $E(S+T)$ is polylogarithmic for prefix normal words. The linear alignment of the data points together with lemma \ref{lemma:random_pnf_st} seems to support that.

%%%%%%%%%%%%%%%%%%%%%%%%%%%%%%%%%%%%%%%%%%%%%%%%%%%%%%%%%%%%%%
%%%%%%%%%%%%%%%%%%%%%%%%%%%%%%%%%%%%%%%%%%%%%%%%%%%%%%%%%%%%%%
%%%%%%%%%%%%%%%%%%%%%%%%%%%%%%%%%%%%%%%%%%%%%%%%%%%%%%%%%%%%%%

\section{Conclusion and Open Problems}

We presented a new generating algorithm for prefix normal words, which produces all prefix normal words of length $n$ in amortized linear  time per word. Notice that the number of words that are {\em not} prefix normal also grows exponentially and greatly dominates prefix normal words (e.g., $\pnn(30)/2^{30}<0.05$), so the gain of any algorithm that runs in amortized time per word, over brute-force testing of all binary words, is considerable. 
%In fact, using our algorithm, we have been able to count the number of prefix normal words up to length $50$ \cite{sloane2}, which would clearly have been impossible using the previous approach. 
We believe, moreover, that our algorithm actually runs in time $O(\polylog (n))$ per word. This could be proved by showing that the expected critical prefix length of a prefix normal word is polylogarithmic in $n$.

In Sec.~\ref{sec:combgenpnw} we gave a linear time testing algorithm for words which are derived from a word $w$ already known to be prefix normal, via a particular operation (swapping the last $1$ of the first 1-run with a $0$ in the first 0-run). This testing algorithm relies both on the knowledge that $w$ is prefix normal, and on the presence of a data structure for $w$ (the $F$-function). We pose as an open problem to find a strongly subquadratic time testing algorithm {\em for arbitrary words}. Another open problem is the computation of prefix normal forms. 
%, i.e.\ the prefix normal word $w'$ which is equivalent to $w$ w.r.t.\ the function defined in Sec.~\ref{sec:pnw}. 
Solving this problem would lead immediately to an improvement for indexed binary jumbled pattern matching.

The observation that our language is a bubble language has opened up
completely new roads. An efficient implementation of the generating algorithm
led to new experimental results which were not available with our previous approach. The obtained
data led to new conjectures and results. We are confident that the connection
to bubble languages will also help in establishing theoretical results about
the number and structure of prefix normal words, and could hopefully lead to a strongly subquadratic testing algorithm.

%
%\begin{small}
%\bibliographystyle{abbrv}
%\bibliography{PNF_generation}

\begin{thebibliography}{10}

\bibitem{AALS03}
A.~Amir, A.~Apostolico, G.~M. Landau, and G.~Satta.
\newblock Efficient text fingerprinting via {Parikh} mapping.
\newblock {\em J. Discrete Algorithms}, 1(5-6):409--421, 2003.

\bibitem{BFKL13}
G.~Badkobeh, G.~Fici, S.~Kroon, and {\relax Zs}.~Lipt{\'a}k.
\newblock Binary jumbled string matching for highly run-length compressible
  texts.
\newblock {\em Inf. Process. Lett.}, 113(17):604--608, 2013.

\bibitem{Benson03}
G.~Benson.
\newblock Composition alignment.
\newblock In {\em Proc. of the 3rd International Workshop on Algorithms in
  Bioinformatics (WABI'03)}, pages 447--461, 2003.

\bibitem{Boecker07}
S.~B{\"o}cker.
\newblock Simulating multiplexed {SNP} discovery rates using base-specific
  cleavage and mass spectrometry.
\newblock {\em Bioinformatics}, 23(2):5--12, 2007.

\bibitem{BoeckerJMS08}
S.~B\"ocker, K.~Jahn, J.~Mixtacki, and J.~Stoye.
\newblock Computation of median gene clusters.
\newblock In {\em Proc.\ of the Twelfth Annual International Conference on
  Computational Molecular Biology (RECOMB 2008)}, pages 331--345, 2008.
\newblock LNBI 4955.

\bibitem{BCFL10}
P.~Burcsi, F.~Cicalese, G.~Fici, and {\relax Zs}.~Lipt{\'a}k.
\newblock On {T}able {A}rrangements, {S}crabble {F}reaks, and {J}umbled
  {P}attern {M}atching.
\newblock In {\em Proc.\ of the 5th International Conference on Fun with
  Algorithms ({FUN} 2010)}, volume 6099 of {\em LNCS}, pages 89--101, 2010.

\bibitem{BCFL12_TOCS}
P.~Burcsi, F.~Cicalese, G.~Fici, and {\relax Zs}.~Lipt{\'a}k.
\newblock On approximate jumbled pattern matching in strings.
\newblock {\em Theory Comput. Syst.}, 50(1):35--51, 2012.

\bibitem{BEL04}
A.~Butman, R.~Eres, and G.~M. Landau.
\newblock Scaled and permuted string matching.
\newblock {\em Inf. Process. Lett.}, 92(6):293--297, 2004.

\bibitem{CFL09}
F.~Cicalese, G.~Fici, and {\relax Zs}.~Lipt{\'a}k.
\newblock Searching for jumbled patterns in strings.
\newblock In {\em Proc.\ of the Prague Stringology Conference 2009 (PSC 2009)},
  pages 105--117. Czech Technical University in Prague, 2009.

\bibitem{CGGLLRT13}
F.~Cicalese, T.~Gagie, E.~Giaquinta, E.~S. Laber, {\relax Zs}.~Lipt{\'a}k,
  R.~Rizzi, and A.~I. Tomescu.
\newblock Indexes for jumbled pattern matching in strings, trees and graphs.
\newblock In {\em Proc.\ of the 20th String Processing and Information
  Retrieval Symposium (SPIRE 2013)}, volume 8214 of {\em LNCS}, pages 56--63,
  2013.

\bibitem{CLWY12}
F.~Cicalese, E.~S. Laber, O.~Weimann, and R.~Yuster.
\newblock Near linear time construction of an approximate index for all maximum
  consecutive sub-sums of a sequence.
\newblock In {\em Proc.\ 23rd Annual Symposium on Combinatorial Pattern
  Matching (CPM 2012)}, volume 7354 of {\em LNCS}, pages 149--158, 2012.

\bibitem{DuhrkopLMB13}
K.~D{\"u}hrkop, M.~Ludwig, M.~Meusel, and S.~B{\"o}cker.
\newblock Faster mass decomposition.
\newblock In {\em WABI}, pages 45--58, 2013.

\bibitem{FL11}
G.~Fici and {\relax Zs}.~Lipt{\'a}k.
\newblock On prefix normal words.
\newblock In {\em Proc.\ of the 15th Intern.\ Conf.\ on Developments in
  Language Theory (DLT 2011)}, volume 6795 of {\em LNCS}, pages 228--238.
  Springer, 2011.

\bibitem{GHLW13}
T.~Gagie, D.~Hermelin, G.~M. Landau, and O.~Weimann.
\newblock Binary jumbled pattern matching on trees and tree-like structures.
\newblock In {\em Proc.\ of the 21st Annual European Symposium on Algorithm
  (ESA 2013)}, pages 517--528, 2013.

\bibitem{GG13}
E.~Giaquinta and {\relax Sz}.~Grabowski.
\newblock New algorithms for binary jumbled pattern matching.
\newblock {\em Inf. Process. Lett.}, 113(14-16):538--542, 2013.

\bibitem{GO80}
L.~J. Guibas and A.~Odlyzko.
\newblock Long repetitive patterns in random sequences.
\newblock {\em Zeitschrift f\"ur Wahrscheinlichkeitstheorie und verwandte
  Gebeite}, 53:241--262, 1980.

\bibitem{KRR13}
T.~Kociumaka, J.~Radoszewski, and W.~Rytter.
\newblock Efficient indexes for jumbled pattern matching with constant-sized
  alphabet.
\newblock In {\em Proc.\ of the 21st Annual European Symposium on Algorithm
  (ESA 2013)}, pages 625--636, 2013.

\bibitem{LLZ12}
L.-K. Lee, M.~Lewenstein, and Q.~Zhang.
\newblock Parikh matching in the streaming model.
\newblock In {\em Proc.\ of 19th International Symposium on String Processing
  and Information Retrieval, SPIRE 2012}, volume 7608 of {\em Lecture Notes in
  Computer Science}, pages 336--341. Springer, 2012.

\bibitem{MR10}
T.~M. Moosa and M.~S. Rahman.
\newblock Indexing permutations for binary strings.
\newblock {\em Inf. Process. Lett.}, 110:795--798, 2010.

\bibitem{MR12}
T.~M. Moosa and M.~S. Rahman.
\newblock Sub-quadratic time and linear space data structures for permutation
  matching in binary strings.
\newblock {\em J. Discrete Algorithms}, 10:5--9, 2012.

\bibitem{Parida06}
L.~Parida.
\newblock Gapped permutation patterns for comparative genomics.
\newblock In {\em Proc.\ of the 6th International Workshop on Algorithms in
  Bioinformatics, (WABI 2006)}, pages 376--387, 2006.

\bibitem{RSW12}
F.~Ruskey, J.~Sawada, and A.~Williams.
\newblock Binary bubble languages and cool-lex order.
\newblock {\em J. Comb. Theory, Ser. A}, 119(1):155--169, 2012.

\bibitem{Ruskey12}
F.~Ruskey, J.~Sawada, and A.~Williams.
\newblock De {Bruijn} sequences for fixed-weight binary strings.
\newblock {\em SIAM Journal of Discrete Mathematics}, 26(2):605--517, 2012.

\bibitem{SW12}
J.~Sawada and A.~Williams.
\newblock Efficient oracles for generating binary bubble languages.
\newblock {\em Electr. J. Comb.}, 19(1):P42, 2012.

\bibitem{sloane2}
N.~J.~A. Sloane.
\newblock The {O}n-{L}ine {E}ncyclopedia of {I}nteger {S}equences.
\newblock Available electronically at \url{http://oeis.org}.
\newblock {S}equence {A}194850.

\bibitem{Wi09}
A.~M. Williams.
\newblock {\em Shift {G}ray {C}odes}.
\newblock PhD thesis, University of Victoria, Canada, 2009.

\end{thebibliography}
%\end{small}

\begin{small}

\end{small}

%%%%%%%%%%%%%%%%%%%%%%%%%%%%%%%%%%%%%%%%%%%%%%%%%%%%%%%%%%%%%%
%%%%%%%%%%%%%%%%%%%%%%%%%%%%%%%%%%%%%%%%%%%%%%%%%%%%%%%%%%%%%%
%%%%%%%%%%%%%%%%%%%%%%%%%%%%%%%%%%%%%%%%%%%%%%%%%%%%%%%%%%%%%%

\section*{Appendix: Connection between prefix normal forms and binary jumbled pattern matching}

\newcommand{\fmax}{F}
\newcommand{\subword}{\sqsubseteq}
\def\pequiv{\sim_P}
\newcommand{\pmax}{\mathop{\rm pmax}}
\newcommand{\pmin}{\mathop{\rm pmin}}
\newcommand{\fmin}{f}
\newcommand{\M}{{\textrm M}}
\renewcommand{\epsilon}{\varepsilon}
\def\P{p}

The linear space solutions for binary pattern matching all rely on a simple property of binary strings, which we refer to as Interval Lemma (folklore)%\footnote{The lemma was first published in~\cite{CFL09}, but it appears to have been known before, i.e.\ folklore.}
: For a binary string $w$ and any fixed length $1\leq k \leq |w|$, if $w$ has two substrings of length $k$, with one containing $x$ $1$s, and the other $y$ $1$s, where $x<y$, then, for any $x\leq z \leq y$, $w$ also contains a substring of length $k$ with exactly $z$ $1$s. In other words, all Parikh vectors of substrings of the same length build an interval. The lemma implies that in order to be able to answer existence jumbled pattern matching queries, it suffices to store, for every length $k$, the maximum and minimum number of $1$s in any substring of length $k$:  When querying whether $w$ has a substring with Parikh vector $(x,y)$, we can simply ask whether $x$ lies between the maximum and minimum number of $1$s for length $x+y$. This list of minima and maxima for every length is the linear size index used. The big open question is how to compute it faster than the current $O(n^2/\log n)$ time. 

Now, prefix normal forms of a word $w$ can be used to compute this index. We know that two words have the same Parikh set (Parikh vectors of substrings) if and only if they have the same prefix normal forms both w.r.t.\ $1$ and to $0$ (see~\cite{FL11}, Thm.~2). 

In Fig. \ref{fig:esempio}, we present the word $w=10011011001001$ and its prefix normal forms in a standard representation for binary words: Draw in the Euclidean plane the word $w$ by representing each letter $1$ by an upper unit diagonal and each letter $0$ by a lower unit diagonal, starting from the origin $(0,0)$. The region between $\PNF_{1}(w)$ and $\PNF_{0}(w)$ forms exactly the Parikh set of $w$. For example, all substrings of length $6$ have one of the Parikh vectors $(4,2),(3,3),(2,4)$. 

\begin{figure}
\begin{minipage}{8cm}
\begin{center}
  \includegraphics[height=35mm]{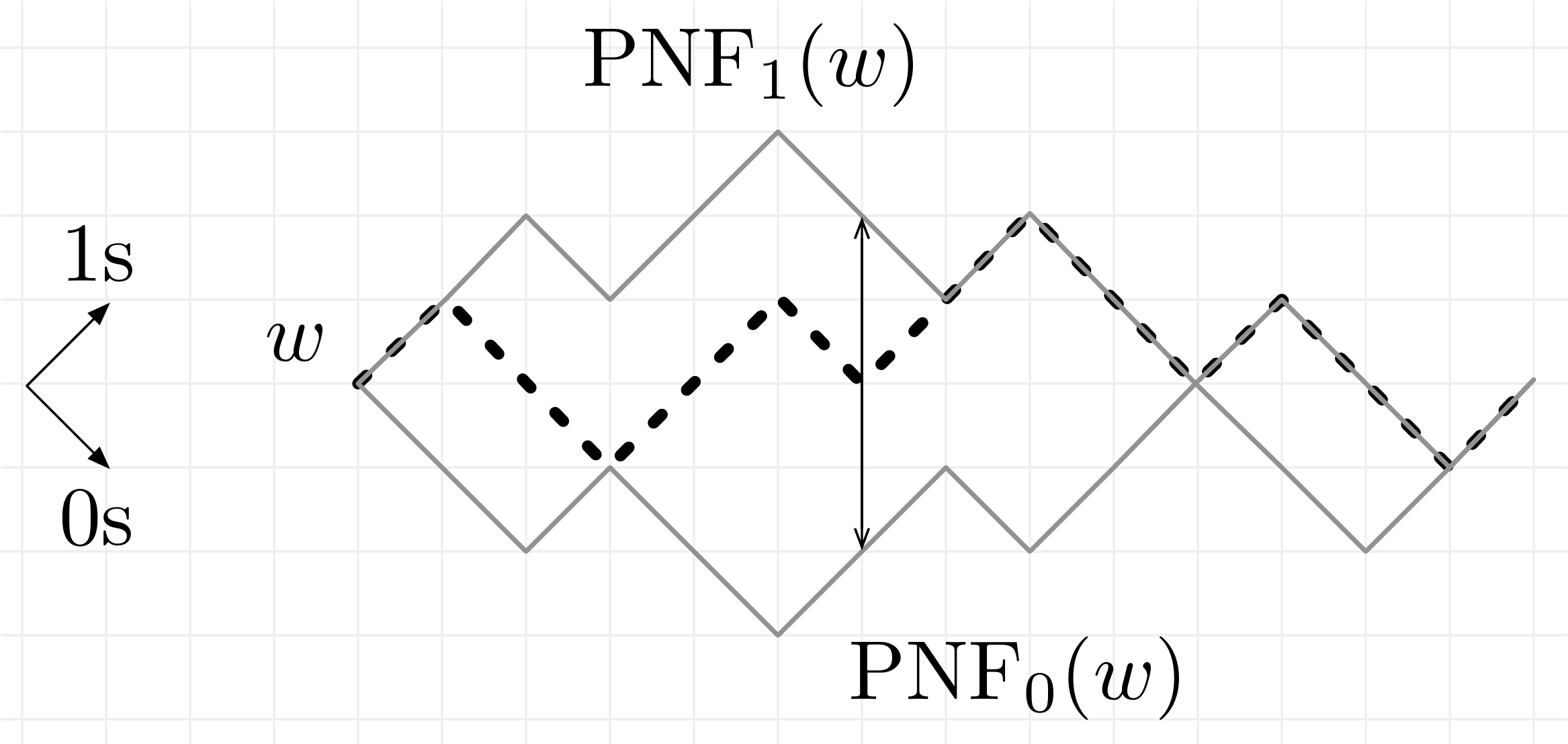}
\end{center}
\end{minipage}
\begin{minipage}{6cm}
\end{minipage}
\begin{minipage}{7cm}
\begin{center}
$
\begin{array}{rcl}
w & = 10011011001001\\
PNF_1(w) & = 11011001001001\\
PNF_0(w) & = 00100110110011
\end{array}
$
\end{center}
\end{minipage}
\caption{The word $w=10011011001001$ (dashed line) and its prefix normal forms (grey lines). The area between the two PNFs is the Parikh set of $w$. The vertical line shows all Parikh vectors of substrings of length $6$, namely $(4,2),(3,3),(2,4)$.}\label{fig:esempio}
\end{figure}

\end{document}